\documentclass{article}
\usepackage{ijcai26}

\usepackage{times}
\usepackage{soul}
\usepackage{url}
\usepackage[hidelinks]{hyperref}
\usepackage[utf8]{inputenc}
\usepackage[small]{caption}
\usepackage{graphicx}
\usepackage{amsmath,amssymb,amsfonts}
\usepackage{amsthm}
\usepackage{booktabs}
\usepackage{algorithm}
\usepackage{algorithmic}
\usepackage[switch]{lineno}
\usepackage{appendix}

\usepackage{subcaption} 
\usepackage{fancyhdr}
\usepackage{float}
\usepackage{algorithm}
\usepackage{algorithmic}
\usepackage{xcolor}
\usepackage{enumitem} 
\usepackage{booktabs}
\usepackage{multirow}
\usepackage{stackengine}
\usepackage{verbatim}

\renewcommand{\algorithmicensure}{\textbf{Output:}}

\newcommand{\cP}{\mathcal{P}}

\newcommand{\fx}{\boldsymbol{x}}

\newtheorem{theorem}{Theorem}

\newtheorem{definition}[theorem]{Definition}

\newtheorem{lemma}[theorem]{Lemma}

\title{Connected EF1 Allocations Exist in Discrete Chore Cutting\thanks{This paper appeared in the proceedings of IJCAI-ECAI 2026. The authors thank Nisarg Shah for bringing to their attention the earlier work of Vittorio Bil{\`{o}}, Martin Loebl, and Cosimo Vinci (AAAI 2026), which independently proved the existence of connected EF1 allocations. In the present paper, we obtain the same result using a different approach.
}}

\author{
Ankang Sun$^1$\and
Bo Li$^2$
\\
\affiliations
$^1$School of Computer Science and Technology, Shandong University\\$^2$Department of Computing, The Hong Kong Polytechnic University\\
\emails
ankang.sun@sdu.edu.cn,
comp-bo.li@polyu.edu.hk
}

\begin{document}

\maketitle

\begin{abstract}
   In this paper, we prove the existence of an envy-free up to one item (EF1) division for a discrete chore.
	     Our approach builds on the powerful framework of Simmons-Su, which leverages Sperner's lemma to guarantee the existence of a simplex corresponding to a sequence of similar fractional divisions, ensuring that each agent is satisfied with a different bundle. Bil{\`{o}} et al. \shortcite{DBLP:journals/geb/BiloCFIMPVZ22}
         introduced a rounding technique that converts the fractional divisions into a connected integral EF1 division for goods
	     when there are at most four agents, and this method was later extended by Igarashi \shortcite{DBLP:conf/aaai/Igarashi23}
         to accommodate any number of agents. 
         {\color{black}However, these rounding techniques for goods do not directly apply to chores because the definitions of EF1 differ in the two settings.
         To overcome this asymmetry, we modify the existing rounding techniques and show that connected EF1 divisions exist for a discrete chore.
         }
\end{abstract}

\section{Introduction}

Cake cutting is a fundamental problem that has attracted extensive attention in economics, mathematics, and computer science \cite{DBLP:books/daglib/0017730,DBLP:reference/choice/Procaccia16,JW}. 
In this problem, the ``cake'' -- represented as the unit interval $[0,1]$ -- serves as a metaphor for a heterogeneous, divisible resource to be allocated among agents with varying preferences.
Beyond its scientific importance, the cake-cutting problem is intellectually engaging and gives rise to a variety of intriguing and challenging questions \cite{DBLP:journals/cacm/Procaccia13}.
One such question is how to divide the cake in an envy-free manner \cite{DBLP:conf/focs/AzizM16,dubins1961cut,VARIAN197463}, i.e., every agent believes her piece is weakly better than all other pieces.
Remarkably, it has been shown in \cite{stromquist1980cut,Su01121999} that an envy-free division, in which every agent receives a contiguous piece, always exists.
The contiguity property is particularly valuable in practical scenarios, especially when the resource has temporal or spatial structures.

The assumption that the cake can be divided arbitrarily is often too strong and unrealistic for many real-world applications. For instance, the use of a lecture room is typically allocated in discrete time slots, and land is commonly divided into distinct plots due to geographical concerns.
These problems fall within the umbrella of the discrete cake-cutting problem \cite{DBLP:journals/tcs/HeydrichS15,DBLP:journals/dam/MarencoT14,DBLP:journals/dam/Suksompong19}, where $m$ indivisible items are arranged on the vertices of a path and each item must be allocated in its entirety to a single agent.
The discrete cake-cutting problem is often considered together with connectivity constraints, which require that each agent receive a contiguous block of items along the path \cite{DBLP:journals/jair/GoldbergHS20,DBLP:journals/iandc/HohneS21,DBLP:conf/atal/KawaseRS24}.

The indivisibility of items introduces significant challenges, since envy-free allocations are not always possible. 
For example, if there is a single indivisible item to be allocated between two agents who both value it positively, the recipient will inevitably be envied by the other. 
To resolve the above impossibility, Bil{\`{o}} et al.
\shortcite{DBLP:conf/innovations/BiloCFIMPVZ19,DBLP:journals/geb/BiloCFIMPVZ22} considered envy-free up to one item (EF1) divisions, first introduced by Budish \shortcite{7c65302b-f079-361a-94f1-0c3c9f6fc76b}. 
In an EF1 division, any agent's envy towards another can be eliminated by removing an item from the latter's allocation.
While EF1 allocations are always guaranteed in the unconstrained setting of the allocations of indivisible items 
\cite{DBLP:conf/sigecom/LiptonMMS04}, the problem becomes more complex in the discrete version of cake-cutting due to the connectivity constraint.
In this paper, we denote a slightly more robust version of EF1 as $\textnormal{EF1}_{outer}$, which informally requires that envy can be eliminated by removing an item at the boundary.
Bil{\`{o}} et al. 
\shortcite{DBLP:conf/innovations/BiloCFIMPVZ19,DBLP:journals/geb/BiloCFIMPVZ22} proved that an $\textnormal{EF1}_{outer}$ division exists for at most four agents, leaving the general case unresolved for several years.
Until 2023, 
Igarashi 
\shortcite{DBLP:conf/aaai/Igarashi23} revised the rounding technique in \cite{DBLP:conf/innovations/BiloCFIMPVZ19,DBLP:journals/geb/BiloCFIMPVZ22} and affirmatively answered the question of the existence of $\textnormal{EF1}_{outer}$ divisions for any number of agents.

The fair allocation of a discrete chore
has been studied in the literature \cite{DBLP:journals/aamas/BouveretCL19,DBLP:journals/iandc/HohneS21}.
{\color{black}Transferring the rounding constructions developed for goods to chores requires accounting for the different definitions of $\textnormal{EF1}_{outer}$ in the two settings.}
This challenge arises from two key differences: first, the allocation of goods and chores is fundamentally distinct; second, the definitions of $\textnormal{EF1}_{outer}$ for goods and chores differ. 
In particular, an $\textnormal{EF1}_{outer}$ division of goods requires that an agent A does not envy agent B after removing a boundary item from B’s bundle. 
In contrast, for chores, the requirement becomes that A does not envy B after removing a boundary item from A’s own bundle. 
This asymmetry in the definitions precludes straightforward reductions between the goods and chores settings, as noted in \cite{DBLP:conf/ijcai/Aziz16}.
{\color{black}In this work, we prove that connected $\textnormal{EF1}_{outer}$ divisions of chores always exist. 
}

\subsection{Our Results and Contributions}\label{sec::resultsandcon}
{\color{black}We prove that} a connected $\textnormal{EF1}_{outer}$allocation of indivisible chores exists for any number of agents with monotone disutility functions.
The discrete version of the cake-cutting problem has been studied by 
Bil{\`{o}} et al. 
\shortcite{DBLP:journals/geb/BiloCFIMPVZ22} and 
Igarashi 
\shortcite{DBLP:conf/aaai/Igarashi23}.
In particular, 
Igarashi 
\shortcite{DBLP:conf/aaai/Igarashi23} provided a bright and powerful framework to achieve $\textnormal{EF1}_{outer}$ for an arbitrary number of agents.
However, the protocol for cake-cutting does not directly work for the chores division problem.
This is because the protocols for positive valuations are not usually applicable for negative valuations, and ``in general there are no reductions from allocation of chores to goods or vice versa'' \cite{DBLP:conf/ijcai/Aziz16}.

In the following, we present the high-level idea of the techniques used in the paper and also compare our techniques with those in \cite{DBLP:journals/geb/BiloCFIMPVZ22,DBLP:conf/aaai/Igarashi23}.
We follow the general framework of Simmons-Su protocol \cite{Su01121999}. 
In particular, it encodes the possible divisions of a path by vertices of a triangulated $(n-1)$-simplex, where $n$ is the number of agents.
It then assigns ownership labels to the vertices of the triangulation, and each agent colors every vertex she owns with the index of her most preferred bundle in the division encoded by that vertex.
By Sperner’s lemma, there exists a fully colored elementary simplex, which corresponds to a sequence of similar divisions in which different agents prefer different bundles. By considering finer and finer triangulations, the protocol yields approximate solutions that converge to an envy-free cake-cutting division.

For the discrete version of cake cutting, Bil{\`{o}} et al. \shortcite{DBLP:journals/geb/BiloCFIMPVZ22} observed that encoding only integral divisions can make the divisions corresponding to the vertices of an elementary simplex too far apart. They therefore considered a finer triangulation by allowing the knife to move in half-steps. 
Moreover, when coloring, agents' preferences over bundles are determined by \emph{virtual valuations},
which do not necessarily coincide with their actual preferences. 
Using half-step triangulations and virtual valuations, Bil{\`{o}} et al. \shortcite{DBLP:journals/geb/BiloCFIMPVZ22} proposed a rounding technique that works for at most four agents.

In follow-up work, Igarashi \shortcite{DBLP:conf/aaai/Igarashi23} pointed out that virtual valuations in \cite{DBLP:journals/geb/BiloCFIMPVZ22} treat left-most and right-most items of a bundle symmetrically (left-right symmetry), which prevents their approach from extending to more than four agents.
Instead, Igarashi \shortcite{DBLP:conf/aaai/Igarashi23} introduced left–right asymmetric virtual valuations and a corresponding rounding technique, 
thereby extending the existence of EF1 divisions in discrete cake-cutting to an arbitrary number of agents.

In this paper, we follow the general framework in \cite{DBLP:journals/geb/BiloCFIMPVZ22,DBLP:conf/aaai/Igarashi23}, but reconstruct key components with new structural insights.
At a high level, the inherent differences between goods and chores, together with the asymmetric definitions of $\textnormal{EF1}_{outer}$ in the two settings, prevent the rounding algorithm of Igarashi \shortcite{DBLP:conf/aaai/Igarashi23} from working for chores.
The detailed differences between our rounding and the existing ones arise from the different definitions of the virtual functions, which act as agents' preferences when coloring.
In particular, the virtual functions in \cite{DBLP:journals/geb/BiloCFIMPVZ22,DBLP:conf/aaai/Igarashi23} rely only on the information between two knife positions.
This limited information is unlikely to suffice for producing an $\textnormal{EF1}_{outer}$ division for chores, regardless of whether the virtual functions exhibit left-right asymmetry.
Instead, we propose new virtual functions that depend not only on the information of the bundle between the knives but also on the item immediately outside to the left of the bundle.
We explain this difference in detail after the definition of our virtual functions in Section~\ref{sec::results}.

Under the framework of Simmons-Su protocol, the procedure of rounding the fully colored elementary simplex is not independent of the coloring functions, which is also the case in \cite{DBLP:journals/geb/BiloCFIMPVZ22,DBLP:conf/aaai/Igarashi23}. 
Since we employ different virtual functions, the detailed procedures for rounding the fully colored elementary simplex into an $\textnormal{EF1}_{outer}$ division inevitably differ from those in \cite{DBLP:journals/geb/BiloCFIMPVZ22,DBLP:conf/aaai/Igarashi23}.
In a nutshell, while we follow the general idea of the framework in \cite{DBLP:journals/geb/BiloCFIMPVZ22,DBLP:conf/aaai/Igarashi23}, we rebuild the virtual functions and the rounding, which are essential for establishing the existence of $\textnormal{EF1}_{outer}$ when dividing a discrete chore among any number of agents with monotone disutility functions.

{\color{black}We note that Bilò et al. \shortcite{DBLP:conf/aaai/BiloLV26} study approximate fairness for non-monotone valuations. In the special case of monotone non-increasing valuations, their Corollary 4.1 establishes the existence of EF1 allocations under path constraints. This result implies the main existence theorem of this paper: a connected $\textnormal{EF1}_{outer}$ allocation of chores exists for any number of agents with monotone disutility functions. 
In the monotone chore setting, any item whose removal eliminates envy must belong to the envious agent's own bundle. Their path-constrained fairness definition requires item removals to preserve connectivity, so the removed item must be a boundary item.
Thus, their path-constrained EF1 guarantee coincides with $\textnormal{EF1}_{outer}$ in our setting.
Our result and proof were developed independently, without knowledge of their work, which was not cited in the conference version of this paper. 
Both proofs use Sperner-type arguments, but the virtual valuations and rounding procedures are different. The present paper therefore provides an independent and alternative proof for the monotone chore setting.}

\subsection{Other Related Work}

The connected cake cutting problem has been extensively studied since the seminal work of 
Stromquist 
\shortcite{stromquist1980cut}, who first proved the existence of a connected EF division.
Su \shortcite{Su01121999} simplified Stromquist's proof via Sperner's lemma, and this technique is now widely used in fair division. 
Aumann et al. 
\shortcite{DBLP:conf/atal/AumannDH13} studied the problem of computing connected divisions that maximize social welfare. 
Bei et al. 
\shortcite{DBLP:conf/aaai/BeiCHTY12} investigated the problem of maximizing social welfare among all proportional connected divisions. 
Aumann and Dombb 
\shortcite{DBLP:journals/teco/AumannD15} quantified the efficiency loss under fairness constraints in connected cake cutting.
For the case of a divisible chore,
Robertson and Webb \shortcite{JW} gave the first envy-free protocol for three agents. Peterson and Su extended this protocol to four agents in \cite{Peterson01042002}, and an arbitrary number of agents in \cite{NpersonSu}.
Dehghani et al. 
\shortcite{DBLP:conf/soda/DehghaniFHY18} gave the first discrete and bounded envy-free protocol for chore division among any number of agents. 
Heydrich and van Stee 
\shortcite{DBLP:journals/tcs/HeydrichS15} studied the trade-off between fairness and efficiency.


The connectivity constraint has also been explored in the discrete version of the cake-cutting problem. 
Marenco and Tetzlaff 
\shortcite{DBLP:journals/dam/MarencoT14} proved the existence of envy-free connected divisions when each item is valued positively by at most one agent. 
Subsequent research has shown the existence of relaxed notions of envy-freeness in cases where items may be valued positively by any number of agents \cite{barrera2017discreteenvyfreedivisionnecklaces,DBLP:journals/dam/Suksompong19}. 
The computational complexity of determining the existence of a fair connected allocation has been studied, for goods in \cite{DBLP:journals/jair/GoldbergHS20,DBLP:conf/atal/KawaseRS24} and for chores in \cite{DBLP:journals/aamas/BouveretCL19}.
The trade-off between fairness and efficiency has been studied in \cite{DBLP:journals/mst/CaragiannisKKK12,DBLP:journals/iandc/HohneS21,DBLP:journals/dam/Suksompong19,DBLP:journals/iandc/SunL25}, considering various definitions of fairness and efficiency.





\section{Preliminaries}
For any positive integer $k$, let $[k] := \{1,\ldots, k\}$. 
There is a set $V$ of $m$ indivisible chores $v_1,\ldots,v_m$ placed on the vertices of a path, which are to be allocated to $n$ agents.
The chores are named from left to right, with chore $v_j$ placed at position $j$.
Throughout the paper, we use $v_j$ both as the identifier of the chore and as its position $j$ on the path, so that $v_j$ can also be compared with real numbers.


A subset $X \subseteq V$ is called \emph{connected} if it induces a connected subgraph. Let $\mathcal{C}(V)$ denote the set of all connected subsets of $V$. We consider the empty set and any singleton set as connected.
Agents are required to receive connected subsets or \emph{bundles}. Moreover, each agent $i \in [n]$ is associated with a cost or disutility function $u_i: \mathcal{C}(V) \rightarrow \mathbb{R}_{\geq 0}$.
Each $u_i$ is \emph{monotone}, that is, for any connected subset $S\subseteq T \subseteq V$, $u_i(S)\leq u_i(T)$.
For ease of representation, we write $u_i(v)$ for $u_i(\{v\})$.

A \emph{division} is a partition $\mathcal{P}=(P_1,\ldots,P_n)$ of the path into $n$ connected bundles, such that for any $i,j\in [n]$, $P_i\cap P_j=\emptyset$ and $\bigcup_{t\in [n]} P_t=V$.
Given partition $\cP$, the subset $P_i$ is the $i$-th piece from the left. 

The division $\cP$ is \emph{envy-free} if there exists a permutation $\pi: [n]\rightarrow [n]$ such that for any $i,j \in [n]$, $u_i(P_{\pi(i)}) \leq u_i(P_{\pi(j)})$ holds.
As we have discussed, envy-free divisions are not guaranteed to exist for indivisible items.
Thus, we shift our attention to the relaxation, namely envy-free up to one item (EF1).
In an EF1 chores division, the envy can be eliminated by removing a chore from the bundle of the envious agent. In the original definition of EF1, there is no connectivity requirement on the resulting bundle after removing an item. 
Motivated by the fact that agents prefer connected subsets over disconnected ones, we instead consider a stronger version, namely \emph{envy-free up to an outer item} ($\textnormal{EF1}_{outer}$), which requires that the resulting bundle after removing an item should be connected.
The notion of $\textnormal{EF1}_{outer}$ has been studied in the context of the discrete version of cake-cutting \cite{DBLP:journals/geb/BiloCFIMPVZ22,DBLP:conf/aaai/Igarashi23}.
Below, we present the formal definition of $\textnormal{EF1}_{outer}$ in the discrete chore division problem.

\begin{definition}[$\textnormal{EF1}_{outer}$]
    For the problem of discrete chore division, a division $\cP=(P_1,\ldots,P_n)$ is $\textnormal{EF1}_{outer}$ if there exists a permutation $\pi:[n]\rightarrow [n]$ such that for any $i,j \in [n]$, either (1) $u_i(P_{\pi(i)}) \leq u_i(P_{\pi(j)})$, 
    or (2) $\exists v\in P_{\pi(i)}$ such that $P_{\pi(i)} \setminus \{v\}$ is connected and $u_i(P_{\pi(i)}\setminus \{v\}) \leq u_i(P_{\pi(j)} )$.
\end{definition}

\subsection{Sperner’s Lemma and Simmons-Su Protocols}

We provide a brief overview of Sperner’s Lemma and Simmons-Su protocols.
Let us begin with combinatorial topology. A $(n-1)$-simplex is the convex hull of $n$ affinely independent points in space $\mathbf{R}^m$ with $m\geq n-1$. 
Let $\textnormal{conv}(\fx_1,\ldots,\fx_n)$ denote the convex hull of the $\fx_1,\ldots,\fx_n$.
For $(n-1)$-simplex $S=\textnormal{conv}(\fx_1,\ldots,\fx_n)$,  $\fx_1,\ldots,\fx_n$ are referred to as the main vertices of $S$.
Given a subset $J \subseteq \{\fx_1,\ldots,\fx_n\}$ with $|J|=k+1$, the $k$-simplex spanned by $J$ is said to be a $k$-face of $S$, denoted as $F_J$.
A \emph{facet} of simplex $S$ is the face spanned by the $n-1$ of the $n$ main vertices of $S$.

A \emph{triangulation} $T$ of an $(n-1)$-simplex $S$ is a collection of smaller simplices $S_1,\ldots,S_t$, satisfying that $\bigcup _{j\in [t]} S_j = S$ and for any $i,j \in [t]$, simplices $S_i$ and $S_j$ either intersect at a shared face or do not intersect at all.
The smaller simplices $S_1,\ldots,S_t$ are the \emph{elementary} simplices of $T$.

Fix simplex $S$ and its triangulation $T$. Let $V(T)$ denote the set of vertices of $T$. In other words, $V(T)$ is the union of the vertices of all the elementary simplices in $T$.
A \emph{coloring} function is a mapping $\lambda: V(T) \rightarrow 2^{[n]}$ that assigns a set of numbers, also referred to as colors, to each vertex of $T$.
A coloring function $\lambda$ is called \emph{Sperner coloring} if the main vertices of $S$ are assigned different colors, and moreover, any other vertex in the interior of some $k$-face must be assigned one of the labels of the main vertices that span that face.
An elementary simplex $S^*=\textnormal{conv}(\fx_1^*,\ldots,\fx_n^*)$ is called \emph{fully colored} under coloring function $\lambda$ if each of its vertices is assigned a distinct color by $\lambda$, i.e., there exists a permutation $\pi:[n]\rightarrow [n]$ such that $\lambda(\fx_i^*) = \pi(i)$ for all $i\in [n]$.



\begin{theorem}[Sperner's lemma]\label{thm::sperner}
    For an $(n-1)$-simplex $S$ and its triangulation $T$, if $V(T)$ is colored by a Sperner coloring function $\lambda$, then there must exist an odd number of fully colored elementary $(n-1)$-simplices.
    In particular, there exists at least one. 
\end{theorem}

\begin{figure*}[htbp]
    \centering
    \begin{subfigure}[b]{0.45\textwidth}
        \centering
        \includegraphics[width=0.8\textwidth]{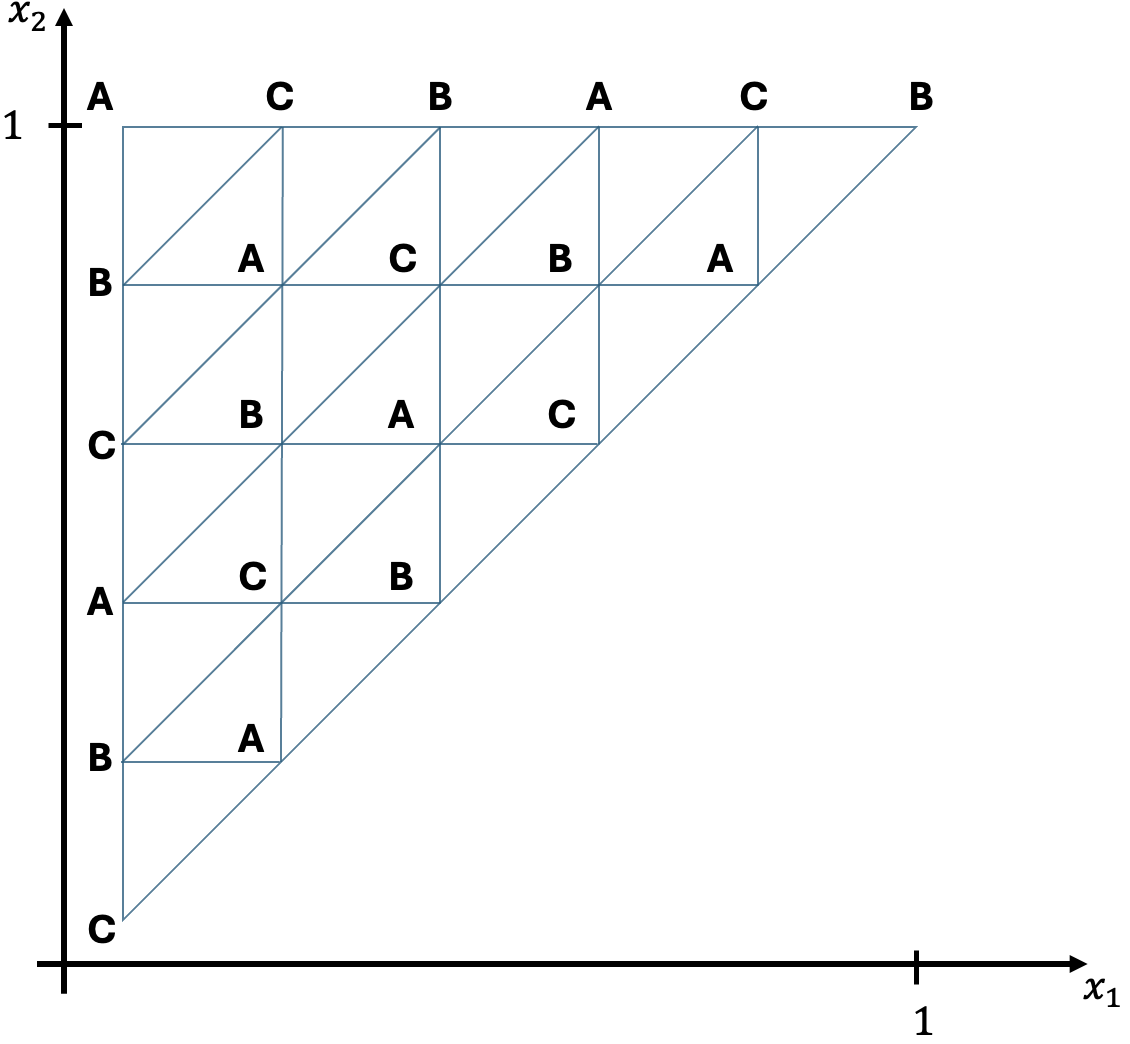} 
        \caption{Owner labeling.}
        \label{fig:subfigA}
    \end{subfigure}
    \begin{subfigure}[b]{0.45\textwidth}
        \centering
        \includegraphics[width=0.8\textwidth]{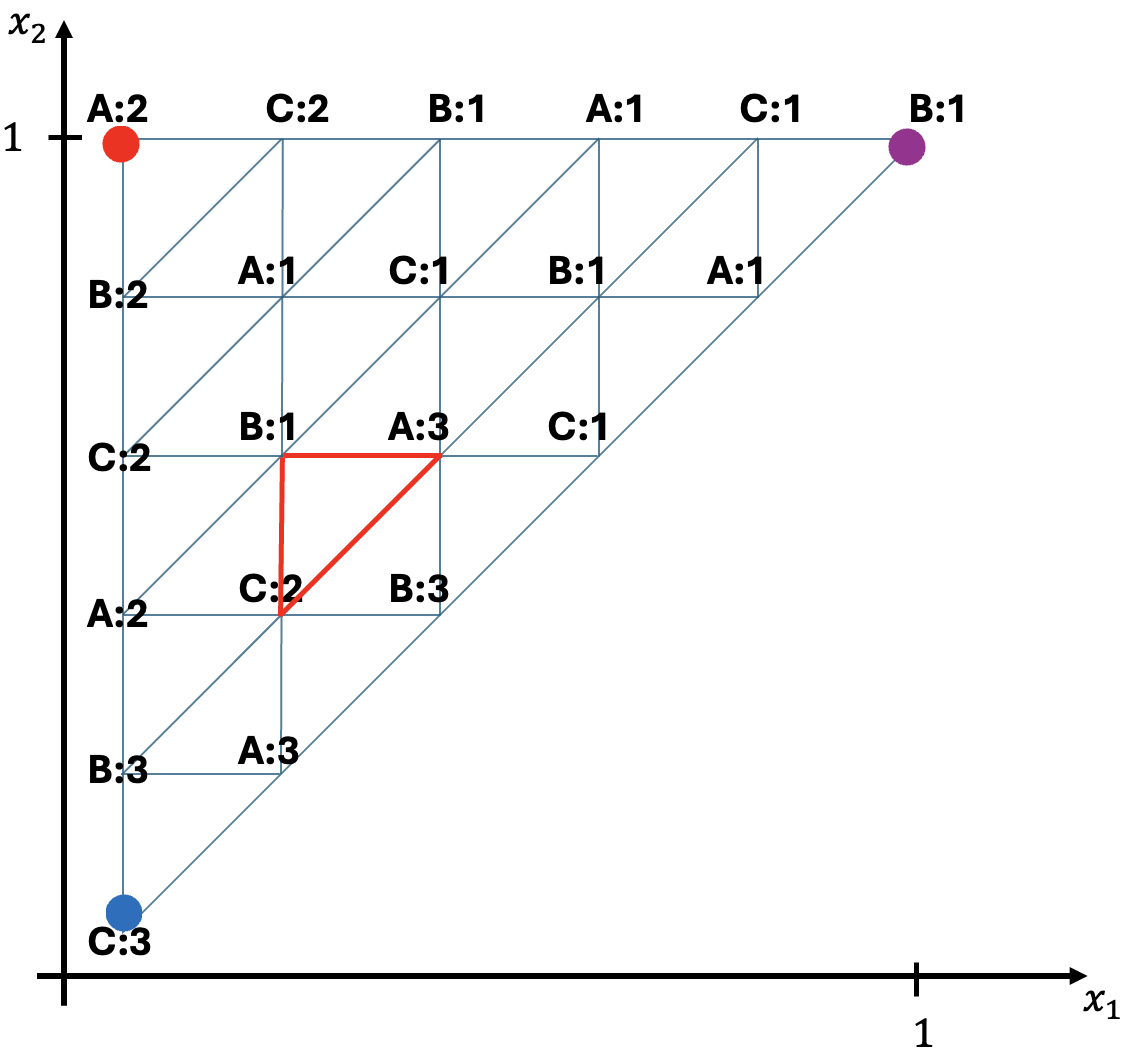} 
        \caption{Coloring}
        \label{fig:subfigB}
    \end{subfigure}


    \caption{Illustration of owner labeling and coloring. The letter A stands for Alice, B for Betty, and C for Charlie. In Figure \ref{fig:subfigB}, the owners of the vertices of the red-marked elementary simplex prefer different pieces, so the simplex is fully colored.}
    \label{fig:mainfigure}
\end{figure*}

We next briefly introduce the Simmons–Su protocols with an example of three agents. Consider dividing a cake $[0,1]$ of total size 1 using 2 knives among Alice, Betty, and Charlie.
Let $x_1$ and $x_2$ be respectively the positions of the first and the second knives. Then $0\leq x_1\leq x_2\leq 1$, and $(x_1,x_2)$ ranges over a 2-simplex $S$ (a triangle).
Then triangulate $S$ with a triangulation $T$ and assign ``ownership'' to each vertex so that the vertices of every elementary simplex in $T$ have three distinct labels, as shown in Figure~\ref{fig:subfigA}.
Each vertex $\fx=(x_1,x_2)$ corresponds to a division of the cake, with $x_i$ being the position of the $i$-th knife.
For example, the upper-left vertex marked in red in Figure \ref{fig:subfigB} is $(0,1)$, and in the corresponding division, the first and third pieces are empty while the second piece is the entire cake.
Visit each vertex of the triangulation and let the owner color the vertex with the indices of their most preferred pieces of the corresponding division.
If each agent prefers any piece with mass over an empty piece, i.e., the agents are \emph{hungry}, then the resulting color is Sperner's color.
By Sperner's lemma, there exists a fully colored elementary simplex $S^*$,
which corresponds to three similar divisions (one for each vertex of $S^*$) in which different agents choose different pieces of cake.
By repeating this procedure on a sequence of finer and finer triangulations, we obtain approximate solutions that converge to an envy-free connected division.

\section{The Existence of $\textnormal{EF1}_{outer}$ Divisions}\label{sec::results}

In this section, we present the main result of this work: for the discrete version of chore division problem, $\textnormal{EF1}_{outer}$ divisions always exist for an arbitrary number of agents with monotone disutility functions.

The discrete version of cake-cutting has been studied by  Bil{\`{o}} et al. 
\shortcite{DBLP:journals/geb/BiloCFIMPVZ22} and  Igarashi \shortcite{DBLP:conf/aaai/Igarashi23}.
They followed the Simmons–Su protocol and rounded a fully colored elementary simplex.
They observed that if the encoded divisions at the vertices are restricted to integral cuts (i.e., the knife always moves in full steps), then the divisions corresponding to the vertices of an elementary simplex are too far apart.
Thus, they considered a finer triangulation that allows the knives to move in half-steps.
In particular, they considered the simplex
$$
S_m= \left\{ \fx \in \mathbb{R}^{n-1}_{+} \mid \frac{1}{2} \leq x^1 \leq \cdots \leq x^{n-1} \leq m+\frac{1}{2} \right\},
$$
and Kuhn's triangulation $\mathrm{T}_{half}$ of $S_m$, with the vertex set $V(\mathrm{T}_{half})$ defined as:
$$
\left\{ \fx \in \mathbb{R}^{n-1}_{+} \mid x^i\in \{\frac{1}{2}, 1, \frac{3}{2}, \ldots, m+\frac{1}{2}\}, \forall i\in [n-1] \right\}.
$$
The Kuhn's triangulation has the property that for each elementary simplex $S=\textnormal{conv}(\fx_1,\ldots,\fx_n)$, there exists a permutation $\phi: [n]\rightarrow [n]$ such that
\begin{equation}\label{eq::kuhn permutation}
       \fx_{\phi(i+1)} = \fx_{\phi(i)} + \frac{1}{2} \boldsymbol{e}^{\phi(i)}, \quad  \forall i \in [n-1],
\end{equation}
where $\boldsymbol{e}^{k}$ denotes the $k$-th standard unit vector (i.e., $k$-th coordinate is 1 and all others are 0).
This property ensures that the vertices of each elementary simplex can be ordered such that each subsequent vertex is generated by adding $\frac{1}{2}$ along a unique coordinate axis, with each axis being selected exactly once over the sequence.
For the purposes of this work, we focus on $S_m$ and Kuhn's triangulation.


For a vertex $\fx=(x^1,\ldots,x^{n-1}) \in V(\mathrm{T}_{half})$, each coordinate $x^i$ functions as a partitioning threshold.
This vertex $\fx$ then induces a (partial) division $\cP(\fx)=(P_1(\fx),\ldots, P_n(\fx))$ of $m$ chores $v_1,\ldots,v_m$, where the subsets are defined as:
\begin{itemize}
    \item $P_1(\fx) = \{v_k \mid v_k < x^1\}$,
    \item for $j=2,\ldots,n-1$:
    $
    P_j(\fx) = \{ v_k \mid x^{j-1} < v_k < x^j \}
    $,
    \item $P_n(\fx) = \{v_k \mid v_k > x^{n-1}\}$.
\end{itemize}
The $\cP(\fx)$ can be a partial division, as not all of $v_1, \ldots, v_m$ are necessarily included in some $P_j(\fx)$.
In particular, if $x^j=v_k$ for some $j$ and $k$, then chore $v_k \notin P_\ell(\fx)$ for all $\ell$.
In this case, we say item $v_k$ is \emph{hidden} by knife $x^j$.


For a vertex $\fx \in V(\mathrm{T}_{half})$ and a bundle $P_j(\fx)$, chore $v_k$ is the \emph{left-boundary item} of $P_j(\fx)$ if $v_k= \lceil x^{j-1} \rceil$. 
We also use $\ell_j(\fx)$ to denote the left-boundary item of $P_j(\fx)$.
Note that $P_1(\fx)$ does not have a left-boundary item, as there is no $x^0$.
Chore $v_k$ is the \emph{right-boundary item} of $P_j(\fx)$ if $v_k= \lfloor x^{j} \rfloor$, and, in this case, we use $r_j(\fx)$ to denote chore $v_k$. 
The $P_n(\fx)$ does not have the right-boundary item, as there is no $x^n$.
We remark that if $\ell_j(\fx)$ is hidden by the $(j-1)$-th knife, then it is the right-boundary item of $P_{j-1}(\fx)$ and the left-boundary item of $P_j(\fx)$.
We say that item $v_k \in V$ \emph{fully appears} in $P_j(\fx)$ if $v_k\in P_j(\fx)$.

Recall that for every elementary simplex $S=\textnormal{conv}(\fx_1,\ldots,\fx_n)$ of $\mathrm{T}_{half}$, there exists a permutation $\phi:[n]\rightarrow [n]$ that satisfies Equation (\ref{eq::kuhn permutation}). 
Then $S$ induces an ordered sequence of partial divisions $\cP(\fx_{\phi(1)}), \cP(\fx_{\phi(2)}), \ldots$, $ \cP(\fx_{\phi(n)})$. 
In this sequence, for $t\geq 2$, each $\cP(\fx_{\phi(t)})$ is obtained from $\cP(\fx_{\phi(t-1)})$ 
by moving one of the knives in half-step, and the movement of each knife occurs exactly once across $n$ partial divisions.
Given elementary simplex $S=\textnormal{conv}(\fx_1,\ldots,\fx_n)$, the path can be decomposed into $(B_1,v^1,B_2,\ldots,v^{n-1},B_n)$, where each $B_j$ is the set of items that fully appear in every $P_j(\fx_t)$, and each $v^j$ is the boundary item that appears either $j$-th bundle or $(j+1)$-th bundle in the sequence of the partial division.
It is crucial to note that as each knife only moves once in a half-step, it never happens that $v^j$ fully appears in the $j$-th bundle of one partial division and fully appears in the $(j+1)$-th bundle of another partial division.
Figure~\ref{fig:kuhn-illustration} illustrates the boundary items, path decomposition, and the property of Kuhn's triangulation.

\begin{figure}[htbp]
    \centering 
    \includegraphics[width=0.48\textwidth]{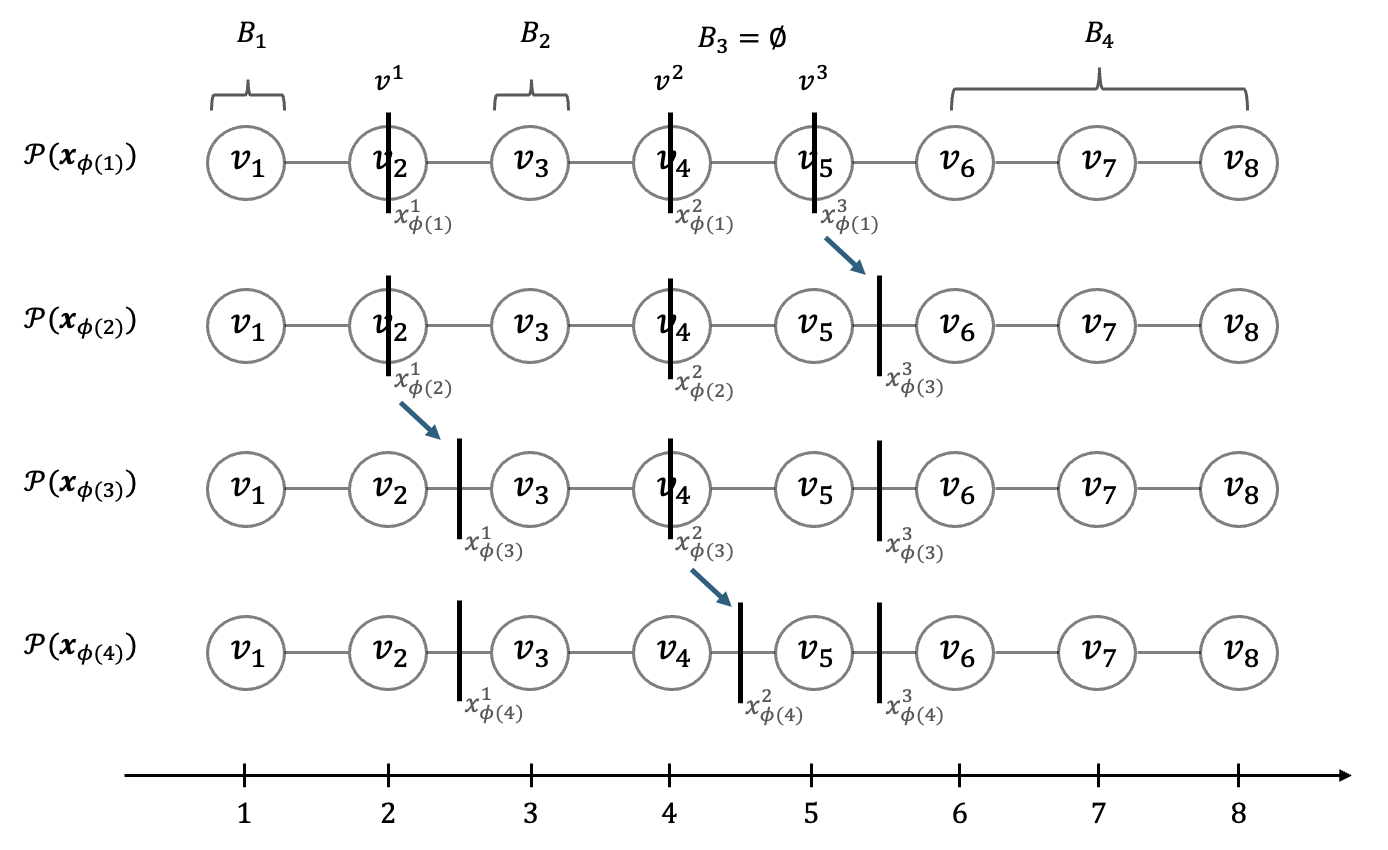} 
    \caption{Illustration of partial divisions of an elementary simplex of Kuhn's triangulation. From $\cP(\fx_{\phi(1)})$, moving the third knife in half-step to the right yields partition $\cP(\fx_{\phi(2)})$.
    Moving the first knife in half-step from $\cP(\fx_{\phi(2)})$ yields $\cP(\fx_{\phi(3)})$, and subsequently moving the second knife in half-step yields $\cP(\fx_{\phi(4)})$.
    Bundle $B_3$ is the empty set as no item fully appears every $P_3(\fx_t)$. For any $t\in [4]$, chore $v_4$ is the right-boundary item of $P_2(\fx_{\phi(t)})$ and is the left-boundary item of $P_3(\fx_{\phi(t)})$.}
    \label{fig:kuhn-illustration} 
\end{figure}

With these preliminary definitions established, we are now prepared to show the existence of  $\textnormal{EF1}_{outer}$ connected division for an arbitrary number of agents with monotone disutility functions.

The first step involves assigning ownership labels to the vertices of $V(\mathrm{T}_{half})$. 
This assignment must ensure that the $n$ vertices within any elementary simplex bear $n$ distinct labels. The existence of such an assignment is guaranteed because $\mathrm{T}_{half}$ is Kuhn's triangulation \cite{DBLP:journals/ior/DengQS12}.



\paragraph{Coloring.}
The second step involves agents coloring the vertices they own.
In Bil{\`{o}} et al. \shortcite{DBLP:journals/geb/BiloCFIMPVZ22} and Igarashi \shortcite{DBLP:conf/aaai/Igarashi23}, the coloring is based on virtual valuations, which are not necessarily identical to agents' valuations.
In this work, we adopt their approach by having agents color vertices based on preferences that differ from their actual disutility functions $u_i$'s.
However, the virtual valuations in \cite{DBLP:journals/geb/BiloCFIMPVZ22,DBLP:conf/aaai/Igarashi23} are designed for cake-cutting and are not applicable to the chore division problem.
Thus, we must define new virtual functions to establish the existence of $\textnormal{EF1}_{outer}$ for chores.

In this work, we define the \emph{virtual disutility} functions $\{\hat{u}_i\}_i$, based on which agents color their vertices. For each agent $i\in [n]$, her virtual function $\hat{u}_i(\fx,j)$ takes as input $\fx\in V(\mathrm{T}_{half})$ and an index $j\in [n]$, returning a non-negative real number.
Informally, $\hat{u}_i(\fx,j)$ represents agent $i$'s virtual disutility for $P_j(\fx)$, and depends on not only the elements in $P_j(\fx)$ but also the boundary items and positions of the $(j-1)$-th and $j$-th knives.


Before formally defining $\hat{u}_i(\fx,j)$, we first introduce the \emph{up-to-one} disutility function $u^-_i:\mathcal{C}(V)\rightarrow \mathbb{R}_{\geq 0}$.
For any $S \in \mathcal{C}(V)$, it is defined as 
\[
u^-_i(S) = \min \left\{ u_i(S\setminus \{v\}) : v\in S \textnormal{ and } S\setminus \{v\} \in \mathcal{C}(V) \right\}.
\]
In particular, when $S$ is the empty set, $u^-_i(S)=0$. 

We now fix agent $i$ and a vertex $\fx$, and formally define virtual functions.
For $j=1$ and $j=n$, $\hat{u}_i(\fx,j) = u_i(P_j(\fx))$. An agent receiving an exterior subset (leftmost or rightmost) does not consider any chore hidden by some knife.


For $j=2,\ldots,n-1$, the virtual function $\hat{u}_i(\fx, j)$ is defined as:
\[
\begin{cases} u_i^{-}\left(P_j(\boldsymbol{x}) \cup \{v^L\} \right) & \text {if } x^{j-1} <\ell_j(\boldsymbol{x}) \text { and } x^j > r_j(\boldsymbol{x}), \\ 
u^-_i\left(P_j(\boldsymbol{x}) \cup \left\{\ell_j(\boldsymbol{x})\right\}\right) & \text {if } x^{j-1} = \ell_j(\boldsymbol{x}) \text { and } x^j> r_j(\boldsymbol{x}), \\ 
u_i\left(P_j(\boldsymbol{x})\right) & \text {if } x^{j-1}<\ell_j(\boldsymbol{x})\text { and } x^j = r_j(\boldsymbol{x}), \\
u_i\left(P_j(\boldsymbol{x})\right) & \text {if } x^{j-1}=\ell_j(\boldsymbol{x}) \text { and } x^j = r_j(\boldsymbol{x}) ,\end{cases}
\]
where $v^L \notin P_j(\fx)$ is the chore to the left of $P_j(\fx)$ such that $\{v^L\}\cup P_j(\fx)$ is connected.
One can introduce a dummy chore $v_0$ immediately to the left of $v_1$, with disutility 0 for every agent, so that $v^L$ always exists.
Moreover, due to Kuhn's triangulation, $x^{j-1}<\ell_j(\fx)$ implies $x^{j-1}= \ell_j(\fx)-\frac{1}{2} $ and $x^j>r_j(\fx)$ implies $x^j=r_j(\fx)+\frac{1}{2}$.

Compared our virtual functions with those in \cite{DBLP:conf/aaai/Igarashi23}, although our $\hat{u}_i(\fx, j)$ handles similar cases, its definitions differ from theirs in each of the four cases.
Moreover, the virtual functions in \cite{DBLP:conf/aaai/Igarashi23} use only information about $P_j(\mathbf{x})$ and its boundary items, while our $\hat{u}_i(\fx, j)$ requires the agent to take $v^L$, an item that is outside $P_j(\mathbf{x})$ and not one of its boundary items




The coloring function of agent $i$ is $\lambda_i : V(\mathrm{T}_{half}) \rightarrow 2^{[n]}$, defined as:
$$
\lambda_i=\arg\min \left\{ \hat{u}_i(\fx,j) : j\in [n] \right\}.
$$
For any $\fx$, there always exists some $j$ achieving the minimum value of $\hat{u}_i(\fx,j)$, so each agent can identify her most preferred bundle. This satisfies the ``Good House'' condition in Su \shortcite{Su01121999}.
We then verify ``Miserly Tenants'' condition in Su \shortcite{Su01121999}, which requires that no agent strictly prefers a piece of positive length to a piece of zero length.
If an agent receives the left or right exterior bundle of zero length, then the corresponding $P_1(\fx)$ or $P_n(\fx)$ is empty, and the virtual disutility is 0.
Moreover, for any $\fx$ and $j=2,\ldots,n-1$ with $x^{j-1}=x^j$, one can verify that $\hat{u}_i(\fx,j)=0$.
The coloring function $\{\lambda_i\}_i$ satisfies ``Miserly Tenants'' condition, and therefore forms a Sperner's coloring.


\paragraph{Rounding.}  
We next introduce how to round an elementary simplex $S=\textnormal{conv}(\fx_1,\ldots,\fx_n)$ into a complete division $\cP^*=(P^*_1,\ldots,P^*_n)$.
Under the framework of Simmons-Su, the rounding procedure depends on the coloring functions, and thus, is not independent of the virtual functions.
Since our virtual functions differ from those in \cite{DBLP:journals/geb/BiloCFIMPVZ22,DBLP:conf/aaai/Igarashi23}, and due to the inherent differences between goods and chores, we must design a new rounding algorithm.



Recall that vertices $\fx_1,\ldots,\fx_n$ decompose the path as $(B_1,v^1,B_2,\ldots,v^{n-1}$, $B_n)$, 
where each $B_j$ is the set of items that fully appear in $P_j(\fx_t)$ for all $t$, i.e., $B_j=\bigcap_{t\in[n]}P_j(\fx_t)$.
The $v^j$'s are items hidden by knives.

Our rounding algorithm is described as follows. For each $j\in [n]$, first initialize $P^*_j$ as $B_j$.
It then remains to allocate the hidden items $v^1,\ldots,v^{n-1}$.
Let us begin with the left-most hidden item $v^1$. If $v^1$ fully appears in the first bundle of some $\fx_t$, i.e., $v^1 \in P_1(\fx_t)$ for some $\fx_t$, then allocate $v^1$ to $P^*_1$.
Next, we allocate $v^j$ with $j=2,\ldots,n-1$ one by one. Suppose that we are at the round of allocating $v^j$. When $v^{j-1}\neq v^j$, if $v^{j-1}$ has not allocated yet, then assign it to $P^*_j$ (Line \ref{step:alg-ef1-unallocate}).
For item $v^j$, assign it to $P_j^*$ if (1) $v^j$ fully appears in the $j$-th bundle of some $\fx_t$ and $v^{j-1}\notin P^*_j$, or (2) $v^j$ fully appears in the $j$-th bundle of some $\fx_t$, $v^{j-1}\in P^*_j$, and there is no $\fx_\ell$ satisfying $x^{j-1}_{\ell} = v^{j-1}$ and $x^{j}_{\ell}  =v^j$ (Line \ref{step:alg-ef1-v^j-condition}).
When $v^{j-1}=v^j$, it trivially holds that $B_j=\emptyset$. Assign $v^j$ to $P^*_j$ if $v^j$ (equivalently $v^{j-1}$) is unallocated and $v^j$ fully appears in the $j$-th bundle of some $\fx_t$ (Line \ref{step:alg-ef1-same}).
The rounding is formally presented in Algorithm \ref{alg:rounding-ef1}.



\begin{algorithm}[h]
	\caption{Rounding an elementary simplex}
	\label{alg:rounding-ef1}
 \renewcommand{\algorithmicensure}{\textbf{Output:}}
	\begin{algorithmic}[1]
		\REQUIRE An elementary simplex $S=\textnormal{conv}(\fx_1,\ldots,\fx_n)$.
		\ENSURE A complete division $\cP^*=(P_1^*,\ldots,P^*_n)$ of the path.
        \STATE For each $j\in [n]$, let $B_j=P_j(\fx_1)\cap P_j(\fx_2) \cap \cdots \cap P_j(\fx_n)$ and let $v^j$ be the item such that $v^j=x^j_t$ for some $t\in [n]$.
        \STATE For each $j\in [n]$, initialize $P_j^* \gets B_j$.
        \IF{item $v^1 \in P_1(\fx_t)$ for some $\fx_t$}
        \STATE $P^*_1 \gets P^*_1\cup\{ v^1 \}$.
        \ENDIF
        \FOR{$j=2,\ldots,n-1$}
        \IF{$v^{j-1} \neq v^j$} 
        \STATE If $v^{j-1}$ is unallocated, i.e., $v^{j-1}\notin \bigcup_{i \in [n]} P^*_i$, 
        then $P^*_j \gets P^*_j \cup\{v^{j-1}\}$.\label{step:alg-ef1-unallocate}
        \STATE Allocate $v^j$ to $P^*_j$ if one of the following two conditions is satisfied: (1) $v^j \in P_j(\fx_t)$ for some $\fx_t$ and $v^{j-1}\notin P^*_j$, or 
        (2) $v^j \in P_j(\fx_t)$ for some $\fx_t$, $v^{j-1}\in P^*_j$, and there is no $\fx_\ell $ such that $x^{j-1}_{\ell} = v^{j-1}$ and $x^{j}_{\ell}  =v^j$. \label{step:alg-ef1-v^j-condition}
        \ENDIF
        \IF{$v^{j-1}=v^j$ and $v^{j-1}$ is unallocated}
        \STATE Allocate $v^j$ to $P^*_j$ if $v^j \in P_j(\fx_t)$ for some $\fx_t$.\label{step:alg-ef1-same}
        \ENDIF
        \ENDFOR
        \IF{$v^{n-1}$ is unallocated}
        \STATE $P^*_n \gets P^*_n\cup \{v^{n-1}\}$.
        \ENDIF
  
	\end{algorithmic}
\end{algorithm}

We next present a crucial lemma for establishing the main result.

\begin{lemma}\label{lem::ef1}
    Let $S=\textnormal{conv}(\fx_1,\ldots,\fx_n)$ be an elementary simplex of $\mathrm{T}_{half}$ and let $\cP^*=(P^*_1,\ldots,P^*_n)$ be the division returned by Algorithm \ref{alg:rounding-ef1} with input $S$.
    For any agent $i$, any bundle $P^*_j$, and any $\fx_k$, it holds that
    $$
    u^-_i(P^*_j) \leq \hat{u}_i(\fx_k,j) \leq u_i(P^*_j).
    $$
\end{lemma}
\begin{proof}[Partial proof]
    We first claim that for each $j\in [n-1]$, $v^{j-1}\leq v^j$. Assume for the contradiction that $v^{j'-1} > v^{j'}$ for some $j' \in [n-1]$.
    Then there exists $k',k''$ such that $x^{j'-1}_{k'}=v^{j'-1} > v^{j'} = x^{j'}_{k''}$. Since $v^{j'-1}$ and $ v^{j'}$ are integers, we have $x^{j'-1}_{k'} \geq x^{j'}_{k''} + 1 > x^{j'}_{k'} $, where the last inequality transition is due to the property of the $\mathrm{T}_{half}$.
    This yields a contradiction. 
    Below, we split the proof based on the value of $j$ and the possible cases of $P^*_j$.

    \paragraph{For $j=1$.} By Algorithm \ref{alg:rounding-ef1}, bundle $P^*_1$ can possibly be $B_1$ or $B_1\cup\{v^1\}$. The arguments below also hold when $B_1=\emptyset$.
    
    \medskip
    \emph{Case 1:} $P^*_1=B_1$. 
    By Algorithm \ref{alg:rounding-ef1}, chore $v^1$ does not fully appear in any $P_1(\fx_t)$, and hence, $x^1_t \in \{ v^1-\frac{1}{2}, v^1\}$ for all $t$.
    By the definition of $\hat{u}_i$, for both cases of $x^1_k=v^1-\frac{1}{2}$ and $x^1_k=v^1$, we have $\hat{u}_i(\fx_k,1)=u_i(B_1)$.
    Therefore, $u^-_i(P^*_1)\leq \hat{u}_i(\fx_k,1) = u_i(P^*_1)$ holds.

    \medskip
    \emph{Case 2:} $P^*_1=B_1\cup\{v^1\}$. By Algorithm \ref{alg:rounding-ef1}, chore $v^1$ fully appears in some $P_1(\fx_t)$, and hence, $x^1_t \in \{v^1, v^1 +\frac{1}{2}\}$ for all $t$.
    Next, we compute possible values of $\hat{u}_i(\fx_k,1)$ and prove that the desired inequalities hold for every possible value.

    (2a): $x^1_k=v^1$. We have $\hat{u}_i(\fx_k,1) = u_i(B_1)$.

    (2b): $x^1_k=v^1+\frac{1}{2}$. We have $\hat{u}_i(\fx_k,1) = u_i(B_1\cup\{v^1\})$.
    \begin{figure}[htbp] 
    \centering 
    \includegraphics[width=0.34\textwidth]{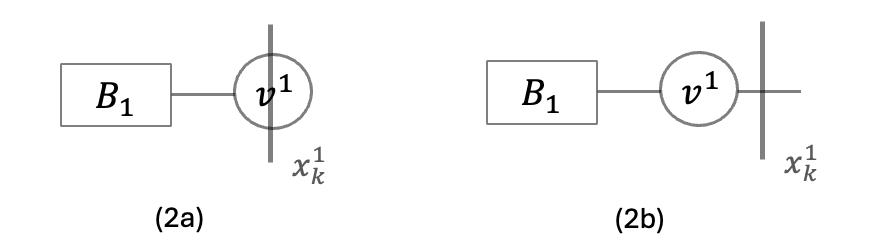} 
    \caption{Illustration of possible cases for $j=1$ and Case 2.}
    \label{fig:j=1-2} 
\end{figure}
    
    We illustrate the two subcases in Figure \ref{fig:j=1-2}. As $v^1$ is the right-most chore of $P^*_1$, we have $u^-_i(P^*_1)\leq u_i(B_1)$. Therefore, it always holds that $u^-_i(P^*_1)\leq \hat{u}_i(\fx_k,1) = u_i(P^*_1)$.
    

    
    \paragraph{For $j=n$.} By Algorithm \ref{alg:rounding-ef1}, bundle $P^*_n$ can possibly be $B_n$ or $\{v^{n-1}\} \cup B_n$. 
    The arguments below also hold when $B_n=\emptyset$.
    
    \emph{Case 1:} $P^*_n=B_n$.
    By Algorithm \ref{alg:rounding-ef1}, chore $v^{n-1}$ fully appears in some $P_{n-1}(\fx_t)$, and hence, $x^{n-1}_{t} \in \{v^{n-1}, v^{n-1}+\frac{1}{2}\}$ for all $t$.
    By the definition of $\hat{u}_i$, for both cases of $x^{n-1}_k=v^{n-1}$ and $x^{n-1}_k=v^{n-1}+\frac{1}{2}$, we have $\hat{u}_i(\fx_k,1)=u_i(B_n)$.
    Therefore, $u^-_i(P^*_n)\leq \hat{u}_i(\fx_k,n)=u_i(P^*_n)$ always holds.

    \emph{Case 2:} $P^*_n=\{v^{n-1}\} \cup B_n$. 
    As $v^{n-1}$ is hidden by some $x^{n-1}_{t'}$, it holds that $x^{n-1}_t \in \{ v^{n-1}-\frac{1}{2}, v^{n-1}, v^{n-1}+\frac{1}{2} \}$ for all $t$ due to the property of Kuhn's triangulation.
    Below, we compute possible values of $\hat{u}_i(\fx_k,1)$ and prove that the desired inequalities hold for all possible values of $\hat{u}_i(\fx_k,1)$.


    (2a): $x^{n-1}_k= v^{n-1}-\frac{1}{2}$. We have $\hat{u}_i(\fx_k,n) = u_i(\{v^{n-1}\} \cup B_n )$.

    (2b): $x^{n-1}_k= v^{n-1}$. We have $\hat{u}_i(\fx_k,n)=u_i(B_n)$.

    (2c): $x^{n-1}_k= v^{n-1}+ \frac{1}{2}$. We have $\hat{u}_i(\fx_k,n)=u_i(B_n)$.

    \begin{figure}[htbp] %
    \centering 
    \includegraphics[width=0.44\textwidth]{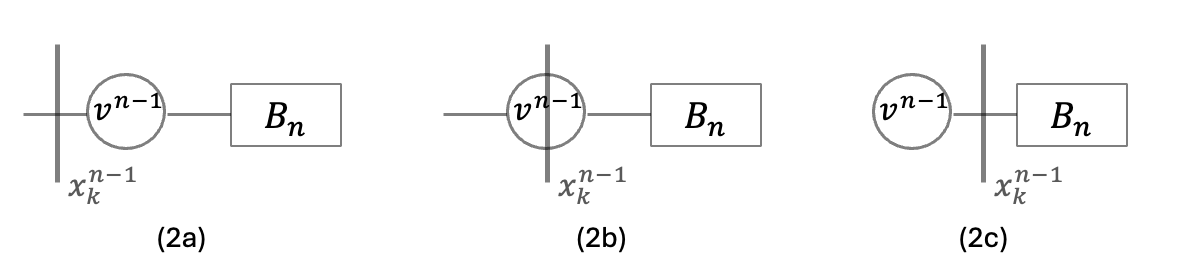} 
    \caption{Illustration of possible cases for $j=n$ and Case 2.}
    \label{fig:j=n} 
\end{figure}

    We illustrate the three subcases in Figure \ref{fig:j=n}.
    As $v^{n-1}$ is the left-most item of $P^*_n$, it holds that $u^-_i(P^*_n)\leq u_i(B_n)$. Since $\hat{u}_i(\fx_k,n)$ is either $u_i(\{v^{n-1}\} \cup B_n )$ or $u_i(B_n)$, the inequality $u^-_i(P^*_n)\leq \hat{u}_i(\fx_k,1) = u_i(P^*_n)$ always holds.

    \paragraph{For $j=2,\ldots,n-1$ and $v^{j-1}<v^j$.}
    By Algorithm \ref{alg:rounding-ef1}, $P^*_j$ can possibly be $B_j$, $B_j\cup\{v^j\}$, $\{v^{j-1}\}\cup B_j$, and $\{v^{j-1}\}\cup B_j \cup \{v^j\}$.
    
    
    \emph{Case 1:} $P^*_j=B_j$. By Algorithm \ref{alg:rounding-ef1}, chore $v^{j-1}$ fully appears in some $P_{j-1}(\fx_{t'})$ and $v^j$ does not fully appear in any $P_{j}(\fx_{t})$.
    Thus, $x^{j-1}_t \in \{v^{j-1}, v^{j-1}+\frac{1}{2}\}$ and $x^j_t \in \{ v^j-\frac{1}{2}, v^j\}$ for all $t$. 
    We compute possible values of $\hat{u}_i(\fx_k,j)$ and prove that the desired inequalities always hold.

    (1a): $x^{j-1}_k=v^{j-1}$ and $x^j_k=v^j-\frac{1}{2}$. We have $\hat{u}_i(\fx_k,j)=u^-_i(\{v^{j-1}\} \cup B_j) $


    (1b): $x^{j-1}_k=v^{j-1}$ and $x^j_k=v^j$. We have $\hat{u}_i(\fx_k,j) = u_i(B_j)$.


    (1c): $x^{j-1}_k=v^{j-1}+\frac{1}{2}$ and $x^j_k=v^j-\frac{1}{2}$. 
    We have $\hat{u}_i(\fx_k,j) = u^-_i(\{v^{j-1} \} \cup B_j)$.


    (1d): $x^{j-1}_k=v^{j-1}+\frac{1}{2}$ and $x^j_k=v^j$. We have $\hat{u}_i(\fx_k,j) = u_i(B_j)$.

    \begin{figure}[htbp] %
    \centering 
    \includegraphics[width=0.48\textwidth]{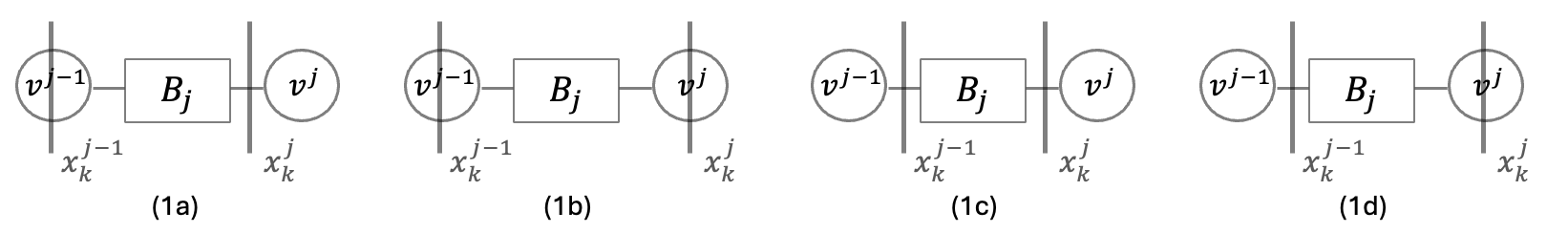} 
    \caption{Illustration of possible cases for $j=2,\ldots,n-1$ and Case 1.}
    \label{fig:J1} 
\end{figure}
    
    We illustrate the four subcases in Figure \ref{fig:J1}. Since $\hat{u}_i(\fx_k,j)$ is either $u^-_i(\{v^{j-1} \} \cup B_j)$ or $u_i(B_j)$, we have $u^-_i(B_j) \leq \hat{u}_i(\fx_k,j)$.
    Moreover, as $v^{j-1}$ is the left-most item of $\{v^{j-1}\}\cup B_j$, we have $u^-_i(\{v^{j-1} \} \cup B_j) \leq u_i(B_j)$.
    Therefore, $u^-_i(P^*_j) \leq \hat{u}_i(\fx_k,j) \leq u_i(P^*_j)$ always holds.

    \emph{Case 2: $P^*_j=B_j\cup\{v^j\}$}. By Algorithm \ref{alg:rounding-ef1}, chore $v^{j-1}$ fully appears in some $P_{j-1}(\fx_{t})$ and $v^j$ fully appears in some $P_j(\fx_{t'})$.
    Thus, $x^{j-1}_t \in \{v^{j-1}, v^{j-1}+\frac{1}{2}\}$ and $x^j_t \in \{ v^j, v^j+\frac{1}{2}\}$ for all $t $.
    We compute possible values of $\hat{u}_i(\fx_k,j)$ and then prove the desired inequalities.

    (2a): $x^{j-1}_k=v^{j-1}$ and $x^j_k=v^j$. We have $\hat{u}_i(\fx_k,j)=u_i(B_j)$.
    

    (2b): $x^{j-1}_k=v^{j-1}$ and $x^j_k=v^j+\frac{1}{2}$. 
    We have $\hat{u}_i(\fx_k,j) = u^-_i(\{v^{j-1}\}\cup B_j \cup \{v^j\})$.
    

    (2c): $x^{j-1}_k=v^{j-1}+\frac{1}{2}$ and $x^j_k=v^j$. We have $\hat{u}_i(\fx_k,j)=u_i(B_j)$.


    (2d): $x^{j-1}_k=v^{j-1}+\frac{1}{2}$ and $x^j_k=v^j+\frac{1}{2}$. We have $\hat{u}_i(\fx_k,j) = u^-_i(\{v^{j-1}\}\cup B_j \cup \{v^j\})$.

    \begin{figure}[htbp] %
    \centering 
    \includegraphics[width=0.48\textwidth]{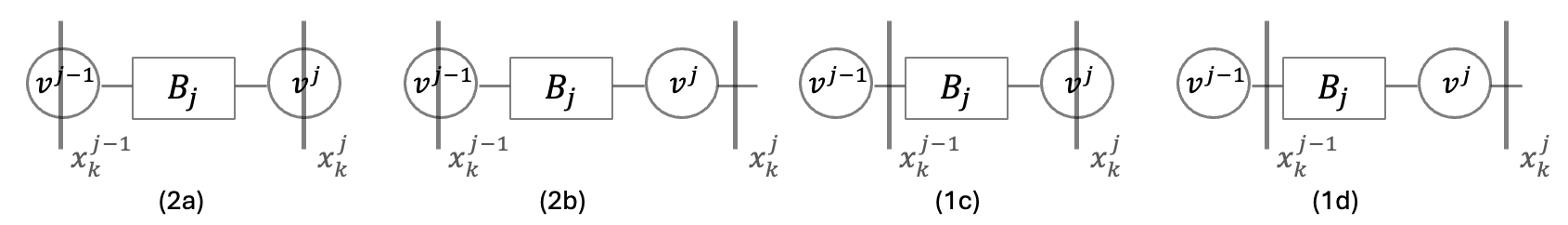} 
    \caption{Illustration of possible cases for $j=2,\ldots,n-1$ and Case 2.}
    \label{fig:J2} 
\end{figure}

    We illustrate the four subcases in Figure \ref{fig:J2}.
    Note that $\hat{u}_i(\fx_k,j)$ is either $u^-_i(\{v^{j-1}\}\cup B_j \cup \{v^j\})$ or $u_i(B_j)$.
    Since $v^j$ is the right-most item of $P^*_j$, it holds that $u^-_i(P^*_j) \leq u_i(B_j)$.
    Moreover, as $u^-_i(\{v^{j-1}\}\cup B_j \cup \{v^j\}) \geq u^-_i(B_j \cup \{v^j\})$, 
    it always holds that $\hat{u}_i(\fx_k,j) \geq u^-_i(P^*_j)$.
    Furthermore, since $v^{j-1}$ is the left-most item of $\{v^{j-1}\}\cup B_j \cup \{v^j\}$, we have $u^-_i(\{v^{j-1}\}\cup B_j \cup \{v^j\}) \leq u_i(B_j\cup \{v^j\})$.
    Therefore, $u^-_i(P^*_j) \leq \hat{u}_i(\fx_k,j) \leq u_i(P^*_j)$ always holds.
    
    \emph{Case 3:} $P^*_j=\{v^{j-1}\}\cup B_j$. For values of $x^{j-1}_k$ and $x^k_j$, we have two situations depending on whether $v^j$ fully appears in some $P_j(\fx_t)$ or not: 
    (1) $v^j$ does not fully appears in any $P_j(\fx_t)$, and (2) $v_j$ fully appears in some $P_j(\fx_{t})$ and there exists some $k'$ such that $x^{j-1}_{k'}=v^{j-1}$ and $x^j_{k'} = v^j$.

    \emph{Case 3.1:} $v^j$ does not fully appears in any $P_j(\fx_t)$. Then, $x^{j-1}_t\in \{v^{j-1}-\frac{1}{2}, v^{j-1}, v^{j-1} +\frac{1}{2} \}$ and $x^j_t\in \{v^j-\frac{1}{2}, v^j\}$ for all $t$.
    We compute possible values of $\hat{u}_i(\fx_k,j)$ and prove the desired inequalities.

    (3.1a) $x^{j-1}_k =v^{j-1}-\frac{1}{2}$ and $x^j_k=v^j-\frac{1}{2}$. We have $\hat{u}_i(\fx_k,j)=u^-_i(\{v^L\} \cup \{v^{j-1}\} \cup B_j )$, where $v^L \notin P_j(\fx_k)$ is the item to the left of $P_j(\fx_k)$ such that $\{v^L\}\cup P_j(\fx_k)$ is connected. In this case, $P_j(\fx_k)=\{v^{j-1}\}\cup B_j$.


    (3.1b) $x^{j-1}_k =v^{j-1}-\frac{1}{2}$ and $x^j_k=v^j$. We have $\hat{u}_i(\fx_k,j)=u_i(\{v^{j-1}\} \cup B_j)$.

    (3.1c) $x^{j-1}_k =v^{j-1}$ and $x^j_k=v^j-\frac{1}{2}$. We have $\hat{u}_i(\fx_k,j)=u^-_i(\{v^{j-1}\} \cup B_j) $.

    (3.1d) $x^{j-1}_k =v^{j-1}$ and $x^j_k=v^j$. We have $\hat{u}_i(\fx_k,j)=u_i(B_j)$. 

    (3.1e) $x^{j-1}_k =v^{j-1}+\frac{1}{2}$ and $x^j_k=v^j-\frac{1}{2}$. We have $\hat{u}_i(\fx_k,j)=u^-_i(\{v^{j-1}\} \cup B_j) $.

    (3.1f) $x^{j-1}_k =v^{j-1}+\frac{1}{2}$ and $x^j_k=v^j$. We have $\hat{u}_i(\fx_k,j)=u_i(B_j)$.

    \begin{figure}[htbp] %
    \centering 
    \includegraphics[width=0.48\textwidth]{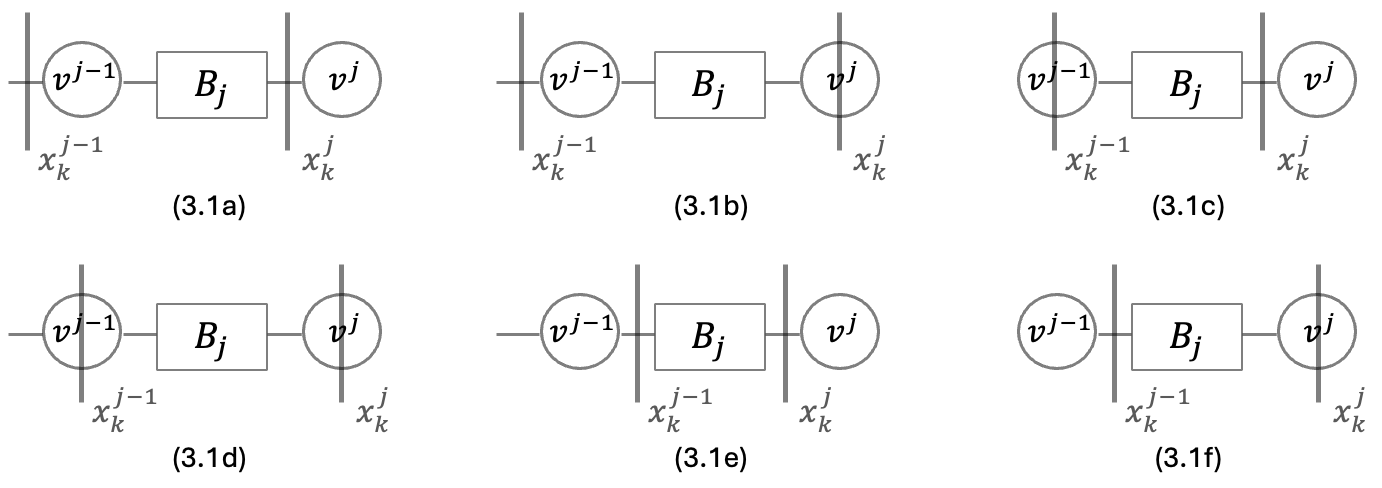} 
    \caption{Illustration of possible cases for $j=2,\ldots,n-1$ and Case 3.1.}
    \label{fig:J3.1} 
\end{figure}

   We illustrate the six subcases in Figure \ref{fig:J3.1}. The possible values of $\hat{u}_i(\fx_k,j)$ are $u_i(B_j)$, $u^-_i(\{v^{j-1}\} \cup B_j)$, $u_i(\{v^{j-1}\} \cup B_j)$, and $u^-_i(\{v^L\} \cup \{v^{j-1}\} \cup B_j )$.
    As $v^{j-1}$ is the left-most item of $P^*_j$, we have $u_i^-(P^*_j) \leq u_i(B_j)$.
    Thus, it is not hard to see that $u_i^-(P^*_j) \leq \hat{u}_i(\fx_k,j)$ always holds. 
    Moreover, as $v^L$ is the left-most item of $\{v^L\} \cup \{v^{j-1}\} \cup B_j$, it holds that $u^-_i(\{v^L\} \cup \{v^{j-1}\} \cup B_j ) \leq u_i(\{v^{j-1}\} \cup B_j)$.
    Therefore, $u^-_i(P^*_j) \leq \hat{u}_i(\fx_k,j) \leq u_i(P^*_j)$ always holds.

    \medskip 

    \emph{Case 3.2:} $v^j$ fully appears in some $P_j(\fx_{t'})$ and there exists some $k'$ such that $x^{j-1}_{k'}=v^{j-1}$ and $x^j_{k'} = v^j$.
    Then $x^{j-1}_t\in \{v^{j-1}-\frac{1}{2}, v^{j-1}, v^{j-1} +\frac{1}{2} \}$ and $x^j_t\in \{v^j, v^j+\frac{1}{2}\}$ for all $t$.
    We compute possible values of $\hat{u}_i(\fx_k,j)$ and prove that the desired inequalities always hold.

    (3.2a) $x^{j-1}_k =v^{j-1}-\frac{1}{2}$ and $x^j_k=v^j$. We have $\hat{u}_i(\fx_k,j)=u_i(\{v^{j-1}\} \cup B_j)$. 

    (3.2b) $x^{j-1}_k =v^{j-1}-\frac{1}{2}$ and $x^j_k=v^j+\frac{1}{2}$. As there exists $k'$ such that $x^{j-1}_{k'}=v^{j-1}$ and $x^j_{k'} = v^j$, the positions of $(j-1)$-th and $j$-th knives of $\fx_k$ can never satisfy the condition of (3.2b); otherwise, violating Equation (\ref{eq::kuhn permutation}).

    (3.2c) $x^{j-1}_k =v^{j-1}$ and $x^j_k=v^j$. We have $\hat{u}_i(\fx_k,j)=u_i(B_j)$.

    (3.2d) $x^{j-1}_k =v^{j-1}$ and $x^j_k=v^j+\frac{1}{2}$. We have $\hat{u}_i(\fx_k,j)=u^-_i(\{v^{j-1}\}\cup B_j \cup \{v^j\})$.

    (3.2e) $x^{j-1}_k =v^{j-1}+\frac{1}{2}$ and $x^j_k=v^j$. We have $\hat{u}_i(\fx_k,j)=u_i(B_j)$.

    (3.2f) $x^{j-1}_k =v^{j-1}+\frac{1}{2}$ and $x^j_k=v^j+\frac{1}{2}$. We have $\hat{u}_i(\fx_k,j)=u^-_i(\{v^{j-1}\}\cup B_j \cup \{v^j\})$.

    \begin{figure}[htbp] %
    \centering 
    \includegraphics[width=0.48\textwidth]{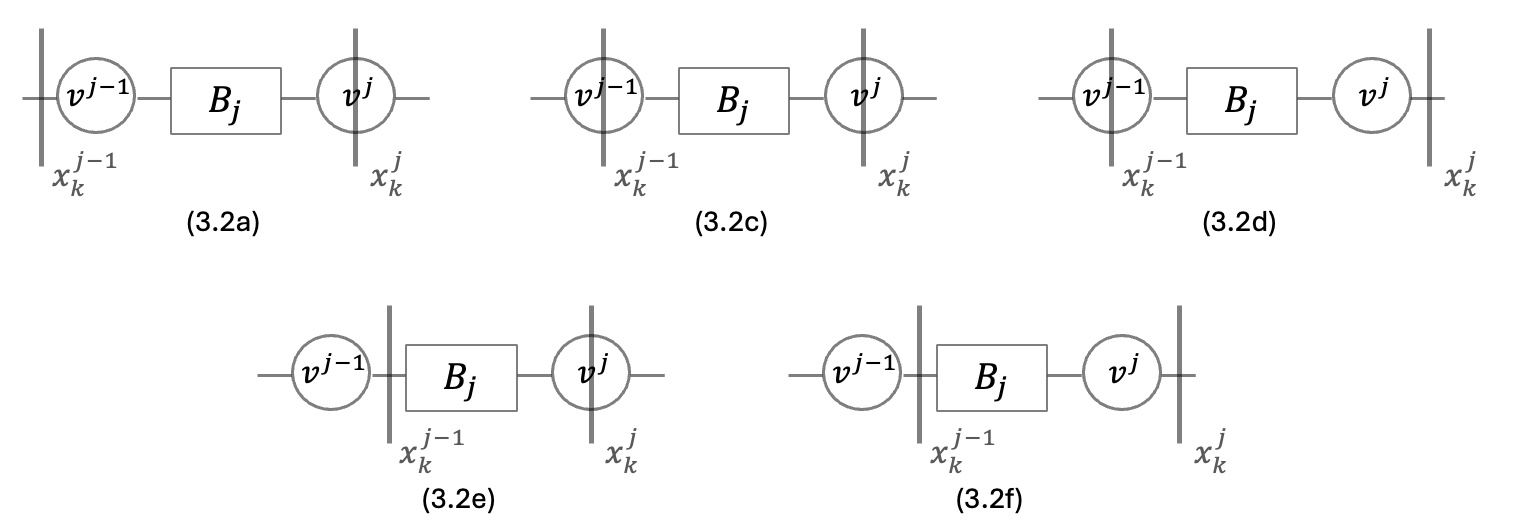} 
    \caption{Illustration of possible cases for $j=2,\ldots,n-1$ and Case 3.2.}
    \label{fig:J3.2} 
\end{figure}

   We illustrate the five possible subcases in Figure \ref{fig:J3.2}. The possible values of $\hat{u}_i(\fx_k,j)$ are $u_i(B_j)$ and $u^-_i(\{v^{j-1}\}\cup B_j \cup \{v^j\})$.
   As $v^j$ is the right-most item of $\{v^{j-1}\}\cup B_j \cup \{v^j\}$, it holds that $u^-_i(\{v^{j-1}\}\cup B_j \cup \{v^j\}) \leq u_i(\{v^{j-1}\}\cup B_j )$.
   Therefore, $u^-_i(P^*_j) \leq \hat{u}_i(\fx_k,j) \leq u_i(P^*_j)$ always holds.

    \medskip 

    \emph{Case 4:} $P^*_j=\{v^{j-1}\} \cup B_j \cup \{ v^j \}$. By Algorithm \ref{alg:rounding-ef1}, chore $v^j$ fully appears in some $P_j(\fx_t)$ and no $\fx_{k'}$ satisfies $x^{j-1}_{k'}=v^{j-1}$ and $x^j_{k'}=v^j$.
    Then $x^{j-1}_t\in \{v^{j-1}-\frac{1}{2}, v^{j-1}, v^{j-1} +\frac{1}{2} \}$ and $x^j_t\in \{v^j, v^j+\frac{1}{2}\}$ for all $t$.
    Next, we compute possible values of $\hat{u}_i(\fx_k,j)$ and prove the desired inequalities.

    (4a): $x^{j-1}_k =v^{j-1}-\frac{1}{2}$ and $x^j_k=v^j$. We have $\hat{u}_i(\fx_k,j)=u_i(\{v^{j-1}\} \cup B_j) $.
    

    (4b) $x^{j-1}_k =v^{j-1}-\frac{1}{2}$ and $x^j_k=v^j+\frac{1}{2}$. We have $\hat{u}_i(\fx_k,j)=u^-_i(\{v^L\} \cup \{v^{j-1}\} \cup B_j \cup \{v^j\} )$, where $v^L \notin P_j(\fx_k)$ is the chore to the left of $P_j(\fx_k)$ such that $\{v^L\} \cup P_j(\fx_k)$ is connected.

    (4c) $x^{j-1}_k =v^{j-1}$ and $x^j_k=v^j$. We note that $\fx_k$ never satisfies the condition of (4c) due to the condition of Case 4.

    (4d) $x^{j-1}_k =v^{j-1}$ and $x^j_k=v^j+\frac{1}{2}$. We have $\hat{u}_i(\fx_k,j)=u^-_i( \{v^{j-1}\} \cup B_j \cup \{v^j\} )$.

    (4e) $x^{j-1}_k =v^{j-1}+\frac{1}{2}$ and $x^j_k=v^j$. We claim that (4e) never happens. Due to Equation (\ref{eq::kuhn permutation}), in order to make $x^{j-1}_k =v^{j-1}+\frac{1}{2}$ and $x^j_k=v^j$ hold, there must be some $\fx_t$ satisfying one of the following: (1) $x^{j-1}_{t}=v^{j-1}$ and $x^j_{t}=v^j$, or (2) $x^{j-1}_{t}=v^{j-1}+\frac{1}{2}$ and $x^j_{t}=v^j-\frac{1}{2}$.
    The condition (1) never holds for any $t$ due to the condition of Case 4, and the condition (2) never holds as $x^j_t\geq v^j$ for all $t$.

    (4f) $x^{j-1}_k =v^{j-1}+\frac{1}{2}$ and $x^j_k=v^j+\frac{1}{2}$. We have $\hat{u}_i(\fx_k,j)=u^-_i(\{v^{j-1}\} \cup B_j \cup \{v^j\} )$.

    \begin{figure}[htbp] %
    \centering 
    \includegraphics[width=0.48\textwidth]{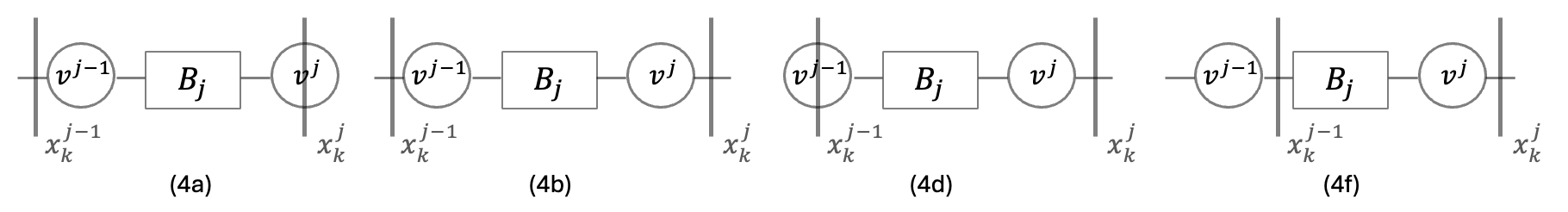} 
    \caption{Illustration of possible cases for $j=2,\ldots,n-1$ and Case 3.2.}
    \label{fig:J4} 
\end{figure}

   We illustrate the four possible subcases in Figure \ref{fig:J4}.
   The possible values of $\hat{u}_i(\fx_k,j)$ are $u_i(\{v^{j-1}\} \cup B_j)$, $u^-_i(\{v^L\} \cup \{v^{j-1}\} \cup B_j \cup \{v^j\} )$, and $u^-_i( \{v^{j-1}\} \cup B_j \cup \{v^j\} )$.
    As $v^j$ is the right-most item of $\{v^{j-1}\} \cup B_j \cup \{v^j\}$, it holds that $u^-_i( \{v^{j-1}\} \cup B_j \cup \{v^j\} ) \leq u_i(\{v^{j-1}\} \cup B_j )$.
    Thus, it is not hard to see that $u^-_i(P^*_j) \leq \hat{u}_i(\fx_k,j)$.
    As $v^L$ is the left-most item of $\{v^L\} \cup \{v^{j-1}\} \cup B_j \cup \{v^j\} $, then $u^-_i(\{v^L\} \cup \{v^{j-1}\} \cup B_j \cup \{v^j\} ) \leq u_i(\{v^{j-1}\} \cup B_j \cup \{v^j\} )$.
    Therefore, $u^-_i(P^*_j) \leq \hat{u}_i(\fx_k,j) \leq u_i(P^*_j)$ always holds.

    \medskip
    \paragraph{For $j=2,\ldots,n-1$ and $v^{j-1}=v^j$.}

    As $v^{j-1}$ and $v^j$ are identical, to avoid confusion, let $v^*$ denote this chore. When $v^{j-1}=v^j$, we have $B_j=\emptyset$, and thus, $P^*_j$ is either the empty set or $\{v^*\}$.

    If $P^*_j=\emptyset$, then $u^-_i(P^*_j)=u_i(P^*_j)=0$. It suffices to show that $\hat{u}_i(\fx_k,j)=0$.
    By the virtual function, if $x^j_k-x^{j-1}_k \leq \frac{1}{2}$, then $\hat{u}_i(\fx_k,j)=0$. 
    Then focus on the case when $x^j_k-x^{j-1}_k > \frac{1}{2}$.
    By definitions of $v^{j-1}$ and $v^j$, there exists $t$ and $t'$ such that $v^{j-1}_t=v^*$ and $v^j_{t'}=v^*$,
    and hence $|x^j_k-v^*|\leq \frac{1}{2}$ and $|x^{j-1}_k-v^*|\leq \frac{1}{2}$.
    Thus, if $x^j_k-x^{j-1}_k > \frac{1}{2}$, then it must be the case that $x^{j-1}_k=v^*-\frac{1}{2}$ and $x^{j}_k=v^*+\frac{1}{2}$.
    As $P^*_j=\emptyset$, 
    chore $v^*$ must be allocated to some $P^*_\ell$ with $\ell < j$. However, as $x^{j-1}_k=v^*-\frac{1}{2}$, we have $x^{\ell}_t\leq v^*$ for all $\fx_t$, that is, $v^*$ does not fully appear in any $P_\ell(\fx_t)$.
    By Algorithm \ref{alg:rounding-ef1}, $v^*$ would not be allocated to $P^*_\ell$, a contradiction.
    Thus, $x^j_k-x^{j-1}_k > \frac{1}{2}$ never happens under the case of $P^*_j=\emptyset$.

    If $P^*_j=\{v^*\}$, then it suffices to show $\hat{u}_i(\fx_k,j) \leq u_i(v^*)$. Again, when $x^j_k-x^{j-1}_k \leq \frac{1}{2}$, then $\hat{u}_i(\fx_k,j)=0$ and we are done.
    If $x^j_k-x^{j-1}_k > \frac{1}{2}$, then it must be the case that $x^{j-1}_k=v^*-\frac{1}{2}$ and $x^{j}_k=v^*+\frac{1}{2}$.
    Then we have $\hat{u}_i(\fx_k,j)=u^-_i(\{v^L\} \cup \{v^*\}) \leq u_i(v^*)$, where $v^L \notin v^*$ is the chore to the left of $v^*$ such that $\{v^L,v^*\}$ is connected.
\end{proof}


\begin{theorem}\label{thm::ef1}
    For an arbitrary number of agents with monotone valuations, there always exists an $\textnormal{EF1}_{outer}$ connected division.
\end{theorem}
\begin{proof}
    First, assign ownership labels to vertices of $\mathrm{T}_{half}$, ensuring that $n$ vertices of any elementary simplex have $n$ distinct labels.
    Then let owners color their vertices based on the $\{\lambda_i\}_i$. 
    By Sperner's lemma (Theorem \ref{thm::sperner}), there exists a fully colored elementary simplex $S^*=\textnormal{conv}(\fx^*_1,\ldots,\fx^*_n)$.
    Without loss of generality, assume that for each $i\in [n]$, vertex $\fx^*_i$ is owned by agent $i$.
    Since $S^*=\textnormal{conv}(\fx^*_1,\ldots,\fx^*_n)$ is fully colored, then there exists a permutation $\pi:[n]\rightarrow [n]$ such that for each $i\in [n]$, the $\pi(i)$-th bundle in the $\cP(\fx_i^*)$ is the bundle most preferred (based on the virtual function) by agent $i$.
    Formally, for any $i,j\in [n]$, it holds that
    $$
        \hat{u}_i(\fx_i^*, \pi(i)) \leq \hat{u}_i(\fx_i^*, j).
    $$
    Then apply Algorithm \ref{alg:rounding-ef1} and round $S^*=\textnormal{conv}(\fx^*_1,\ldots,\fx^*_n)$ to the complete division $\cP^*=(P^*_1,\ldots,P^*_n)$.
    For an agent $i \in [n]$, bundle $P^*_{\pi(i)}$, vertex $\fx^*_i$, and $j\in [n]$, we have
    \[
    \begin{aligned}
        u^-_i(P^*_{\pi(i)}) & \overset{(a)}{\leq}  \hat{u}_i(\fx^*_i,\pi(i))
        \overset{(b)}{\leq} \hat{u}_i(\fx^*_i,j) 
        \overset{(c)}{\leq} u_i(P^*_{j}),
    \end{aligned}
    \]
    where (a) and (c) follow from Lemma~\ref{lem::ef1} and (b) is because agent $i$ color her vertex $\fx^*_i$ as $\pi(i)$.
    Therefore, assigning each $i$ bundle $P^*_{\pi(i)}$ yields an $\textnormal{EF1}_{outer}$ connected division.
\end{proof}

\section{Conclusion}
In this work, we prove that connected $\textnormal{EF1}_{outer}$ divisions of indivisible chores exist for an arbitrary number of agents with monotone disutility functions. 
On top of Simmons-Su's framework, we introduce a rounding approach which relies on virtual disutility functions that utilize not only the information of the bundle between the boundaries but also the information outside.

Our results raise many interesting questions for future research. 
We focus on fairness in this work, while exploring the compatibility of fairness and Pareto optimality (PO) is also an interesting direction. 
In particular, 
Mahara \shortcite{DBLP:conf/soda/Mahara26} demonstrated that EF1 and PO are compatible when all items are chores (without connectivity constraints) and agents have additive disutility functions; that is, the disutility of a subset equals the sum of the disutilities of all chores within it.  
It is worthwhile to investigate whether EF1 and PO are also compatible in the discrete chores division setting.  
In addition, our work provides only a constructive proof, and the computational complexity of finding an $\textnormal{EF1}_{outer}$ chores division remains unknown.






\bibliographystyle{named}
\bibliography{ijcai26}

\end{document}